\documentclass[11pt]{article}

\usepackage[utf8]{inputenc}
\usepackage[margin=1in]{geometry}
\usepackage{amsmath,amssymb,amsthm,amsfonts,mathrsfs}
\usepackage{graphicx}
\usepackage{booktabs}
\usepackage{siunitx}
\usepackage{float}

\usepackage{caption}
\usepackage{subcaption}
\usepackage{multirow}
\usepackage{algorithm}
\usepackage[noend]{algpseudocode}
\usepackage{tikz}
\usetikzlibrary{positioning,arrows.meta,calc}
\usepackage{physics}
\usepackage{dsfont}
\usepackage{fancyhdr}
\usepackage{xcolor}
\definecolor{plantnode}{RGB}{180,210,170} 
\definecolor{edgegray}{RGB}{90,90,90}     
\usepackage{ragged2e}
\usepackage{hyperref}
\usepackage[
  backend=biber,
  bibstyle=ieee,
  citestyle=numeric-comp, 
  sorting=none,           
  maxnames=3,             
  minnames=1             
]{biblatex}
\newtheorem{theorem}{Theorem}
\newtheorem{lemma}[theorem]{Lemma}

\title{\begin{center}
\textbf{UNDERSTANDING AND MANAGING FROGEYE LEAF SPOT THROUGH\\
NETWORK-BASED MODELING IN SOYBEAN}\\[4pt]
{\large \textit{A PREPRINT}}\\[1pt]

\end{center}}
\usepackage{tabularx} 

\author{
\begin{tabularx}{\textwidth}{@{}*{2}{>{\centering\arraybackslash}X}@{}}
\textbf{Chinthaka Weerarathna} & \textbf{Thien-Minh Le$^{*}$} \\[6pt]
Department of Mathematics & Department of Mathematics \\ 
University of Tennessee at Chattanooga & University of Tennessee at Chattanooga \\ 
Chattanooga, TN & Chattanooga, TN \\[6pt]
\texttt{chinthakaw236@gmail.com} & \texttt{thien-le@utc.edu}
\end{tabularx}
\\[60pt]
\textbf{Jin Wang}\\[6pt]
Department of Mathematics\\
University of Tennessee at Chattanooga\\
Chattanooga, TN\\[6pt]
\texttt{jin-wang02@utc.edu}
}

\date{\today}

\begin{document}
\maketitle

\begingroup
\renewcommand\thefootnote{}
\footnotetext{$^{*}$corresponding author: Thien-Minh Le; thien-le@utc.edu}
\endgroup

\begin{abstract}

\noindent
Frogeye Leaf Spot (FLS), caused by \emph{Cercospora sojina}, poses a significant threat to soybean production, with yield losses of 30--60\%. Traditional mass-action models assume homogeneous mixing, which rarely holds in real fields and limits their ability to inform FLS management. To address this, we developed a network-based model that incorporates real-field structure to improve FLS management in soybeans. Using approximate Bayesian computation, we estimated key epidemiological parameters and found that infection origin can shift the balance between transmission routes. Data analyses indicated that tillage and non-tillage plots did not differ significantly in fungal spread, decay, or disease severity. Finally, we show that early, targeted roguing is more effective than delayed or random removal. Together, these findings offer science-based guidance for FLS management and highlight the value of network-based models to inform agricultural disease control.
\end{abstract}



\section{Introduction} \label{sec:intro}
Frogeye leaf spot (FLS) is a widely distributed soybean disease attributed to the fungus \textit{Cercospora sojina}~\cite{Barro2023FLS}. In recent years, the disease has been reported with increasing frequency in regions that historically experienced lower FLS pressure, including several Midwestern and Northern soybean producing states~\cite{Barro2023FLS,Cai2024FLS}. This broader distribution is partly associated with warmer winter conditions and the extended survival of infected residue at the soil surface, where the pathogen may persist for up to two years in conservation tillage systems~\cite{Barro2023FLS,Neves2022,NCStateFLS,NDSU_FLS2023}. At the same time, reduced sensitivity of \textit{C. sojina} to commonly used fungicides has been documented across multiple soybean producing regions, decreasing the reliability of single mode products and motivating more integrated disease management strategies~\cite{Neves2020QoI,Neves2021Baseline,Bradley2023INCCA}.

FLS epidemics can lead to significant yield reductions by diminishing the photosynthetic area and causing early defoliation. Untreated plot studies have reported yield losses of approximately $10-20\%$, while under highly favorable environmental conditions and susceptible cultivars, losses can reach $30-60\%$~\cite{Barro2023FLS,agronomy_FLS_Seviarity,Dashiell1991}. At larger scales, coordinated disease surveys from the United States and Ontario show that FLS contributed to annual yield reductions of approximately {$10^{5}$ to $10^{6}$} metric tons between 2010 and 2019, with losses during 2015 - 2019 generally exceeding those observed in 2010 - 2014~\cite{Allen2017,Bradley2021}. Recent evidence indicates that FLS remains a growing and geographically expanding threat to soybean production in the United States. Although historically considered a predominantly southern disease, FLS has become increasingly prevalent in northern soybean-growing regions, with recent outbreaks reported across the Midwest, including Indiana, Iowa, Ohio, Wisconsin, Illinois, and North Dakota~\cite{Cai2024IndianaFLS,Neves2022ND,CPN2024}. In North Dakota, where FLS was first confirmed in 2020, continued field scouting has documented its presence across dozens of counties, highlighting ongoing northward expansion~\cite{Webster2026NDReport}. One important factor contributing to the recent expansion of the disease is the widespread development of resistance to commonly used foliar fungicides. Since 2010, fungicide-resistant populations of \textit{C. sojina} have been detected across more than 20 soybean-producing states, reducing the effectiveness of chemical control and increasing the risk of disease spread~\cite{Zhang2012,Neves2022}. In the 2024 growing season alone, FLS caused an estimated \$30.9 million in soybean yield losses in the United States, with the greatest economic impacts concentrated in southern and central production regions~\cite{CPN2024}. Together, these trends indicate that FLS is increasing both in economic impact and geographic spread, highlighting the need for modeling frameworks that account for localized transmission, spatial heterogeneity, and field-level management practices.

FLS involves both direct plant-to-plant transmission and an environmental pathway linked to inoculum that persists in residue and soil~\cite{Barro2023FLS, MBE2024FLS}. Infections initiated early in the season can lead to repeated cycles of secondary spread under favorable weather, giving the epidemic a strongly polycyclic character~\cite{Wise2015}. These features highlight the need for models that represent both local plant interactions and environmental contributions to transmission.

Mathematical epidemiology provides a structured way to describe how infections emerge, spread, and decline in host populations, and how management actions can alter those trajectories. Compartment models, introduced by Kermack and McKendrick in their classical formulation of the Susceptible–Infectious–Removed (SIR) system~\cite{KermackMcKendrick1927}, remain the foundation for many mechanistic approaches. Extensions such as the Susceptible–Exposed–Infectious–Removed (SEIR) model incorporate latent periods that are essential for accurately representing plant pathogens, whose progression from infection to infectiousness often spans several days. These frameworks allow key epidemiological quantities such as thresholds for invasion, equilibrium states, and the basic reproduction number \(R_{0}\) to be derived using dynamical systems theory~\cite{Murray2003,SIRModel}. Although widely applied in human and animal epidemiology, mechanistic compartment models have been used far less extensively in plant pathology~\cite{Garrett2018Networks,ShawPautasso2014}. Previous mathematical efforts to study FLS have used empirical, statistical, and mechanistic approaches. Empirical studies based on regression analyses and disease progress metrics such as area under the disease progression curve (AUDPC) quantify severity and yield loss but do not represent the underlying transmission process~\cite{Barro2023FLS,ISU_FLS}. Mechanistic modeling of FLS has recently been initiated by Yang and Wang, who proposed the first ordinary differential equation (ODE) framework developed specifically for soybean to investigate FLS epidemics~\cite{MBE2024FLS}
. Their model explicitly incorporates both primary and secondary transmission pathways and couples plant infection dynamics with the intrinsic accumulation and decay of pathogen inoculum in contaminated soil, supported by equilibrium, stability, and numerical analyses. More recently, this framework has been extended to partially diffusive partial differential equation (PDE) models that incorporate spatial diffusion of the pathogen within soybean fields, enabling characterization of the resulting spatiotemporal disease dynamics~\cite{PDE_Jin_2025}. In parallel, fractional-order extensions have been proposed to capture memory effects and historical dependence in FLS transmission, providing additional flexibility for modeling long-term disease persistence~\cite{FARMAN}. Additional studies have employed meta-analytic methods to quantify severity-yield relationships and machine learning approaches for disease detection and classification~\cite{agronomy_FLS_Seviarity,agronomy_machine_learning}. These approaches are formulated in empirical or data-driven frameworks and do not explicitly represent transmission mechanisms or infection progression at the plant level. Existing mechanistic models of FLS are primarily deterministic and are developed under the assumption of homogeneous mixing or continuous spatial diffusion; however, they do not account for discrete plant-to-plant contact structure, field geometry, or localized residue-mediated inoculum. Consequently, these modeling approaches limit their ability to study the effects of management interventions such as tillage or roguing within a spatially explicit, plant-level modeling framework calibrated to field observations.

To address the existing gap, we propose a network-based SEIRB model that represents individual soybean plants as nodes connected by proximity-defined edges that govern transmission. This formulation replaces homogeneous mixing with an explicit plant-level contact structure informed by field geometry and is coupled to a dynamic soil reservoir that captures both plant-to-plant and soil-mediated infection pathways. The model is calibrated to field observations using approximate Bayesian computation, enabling statistical inference of transmission and environmental parameters and explicit quantification of uncertainty. Within this unified framework, we examine how alternative initial spatial configurations influence epidemic trajectories and evaluate the impacts of management interventions, including tillage and roguing, on residue-driven inoculum, network connectivity, and disease outcomes.

The remainder of this article is organized as follows. Section~\ref{sec:methods} describes the methodology. In particular, Section~\ref{sec:priliminaries} introduces the classical SEIRB formulation of Yang and Wang~\cite{MBE2024FLS} for FLS, while Section~\ref{sec:Network_model} develops the proposed network-based model as a spatial extension of the classical framework. Section~\ref{sec:results} reports the main findings. Specifically, Section~\ref{sec:Res_data} presents parameter estimation results; Section~\ref{sec:res_tillage} evaluates the effects of tillage on direct transmission, soil-mediated exposure, and cumulative disease outcomes; and Section~\ref{sec:res_rogue} investigates the effectiveness of roguing strategies under different protocols within the spatial contact network. The Appendix provides additional analytical insight into the equilibrium behavior of the soil reservoir. Finally, Section~\ref{sec:conclusion} summarizes the main findings, discusses implications for disease management, and outlines future research directions.

\section{Methods} \label{sec:methods}

We begin by reviewing the SEIRB compartment model of Yang and Wang~\cite{MBE2024FLS} (Section~\ref{sec:priliminaries}), and then present its network-based extension as the proposed model (Section~\ref{sec:Network_model}).

\raggedbottom\subsection{The SEIRB Framework for Frogeye Leaf Spot by~\cite{MBE2024FLS} }\label{sec:priliminaries}

Compartmental models of the Susceptible-Exposed-Infectious-Recovered (SEIR) type have been adapted in many forms for plant disease systems, with modifications introduced to represent latent periods, environmental reservoirs, or multiple transmission routes. For FLS, the Susceptible-Exposed-Infectious-Recovered-Bacteria (SEIRB) formulation developed by Yang and Wang~\cite{MBE2024FLS} extends the classical SEIR framework by incorporating an environmental compartment \(B(t)\) that tracks inoculum present in soil and crop residue. In this formulation, each plant occupies one of four epidemiological states: susceptible \(S(t)\), exposed \(E(t)\), infectious \(I(t)\), or removed \(R(t)\), while the environmental compartment \(B(t)\) represents the density of active fungal spores in soil or plant debris. Susceptible plants become exposed through two transmission pathways: direct contact with infectious plants and soil-mediated inoculum (see Fig.~\ref{fig:seirb_and_network}a). The resulting dynamical system is given by system~(\ref{eq:siroriginal}) below.

\begin{equation}\label{eq:siroriginal}
\left\{\begin{array}{llll}
\dfrac{\mathrm{d} S(t)}{\mathrm{d} t} = \mu_{0} N- (\theta I+ \beta B)S-\mu_{0} S  \\ \\
\dfrac{\mathrm{d} E(t)}{\mathrm{d} t} = (\theta I+ \beta B)S- \mu_{0} E - \sigma_{0} E\\ \\
\dfrac{\mathrm{d} I(t)}{\mathrm{d} t} = \sigma_{0} E - \mu_{0} I - \gamma_{0} I \\ \\
\dfrac{\mathrm{d} R(t)}{\mathrm{d} t} = \gamma_{0} I -\mu_{0} R \\ \\
\dfrac{\mathrm{d} B(t)}{\mathrm{d} t} = r_{0}B(1-\dfrac{B}{k_{0}})-\tau B +\xi I \\ \\

\end{array}\right.
\end{equation}

Here \(N\) is the number of total plants (constant). The parameter \(\mu_{0}\) is the natural birth and removal rate of plants, \(\theta\) and \(\beta\) are the secondary and primary transmission rates, \(\sigma_{0}^{-1}\) is the mean latent period, and \(\gamma_{0}\) is the removal rate of infectious plants. The environmental inoculum dynamics are governed by the intrinsic growth rate \(r_{0}\), carrying capacity \(k_{0}\), decay rate \(\tau_{0}\), and contribution rate \(\xi\) from infectious plants. All parameters are nonnegative.

The above compartmental SEIRB model was formulated to describe FLS transmission through both plant-to-plant and soil-mediated pathways under the assumption of homogeneous mixing. The biological progression parameters were fixed at $\sigma_{0}=1/10$ and $\gamma_{0}=1/75$ per day, with natural plant removal $\mu_{0}=1/150$ per day. Soil inoculum dynamics followed logistic growth with $r_{0}=0.001$ and $k_{0}=60{,}000$, and the environmental decay rate was set to $\tau_{0}=1/(2\times365)$ per day based on empirical evidence of pathogen persistence in crop residue. The initial soil inoculum density is $B(0)=4000$ conidia per ml. These fixed parameter values were adapted from previously published experimental and modeling studies to ensure biological consistency with established disease progression dynamics. In this work, the authors estimated the transmission parameter vector $\boldsymbol{\psi}=(\theta,\beta,\xi)$ using field disease-severity data. This formulation provides the mechanistic foundation for the spatial network extension developed in the next section.

\vspace{0.8cm}

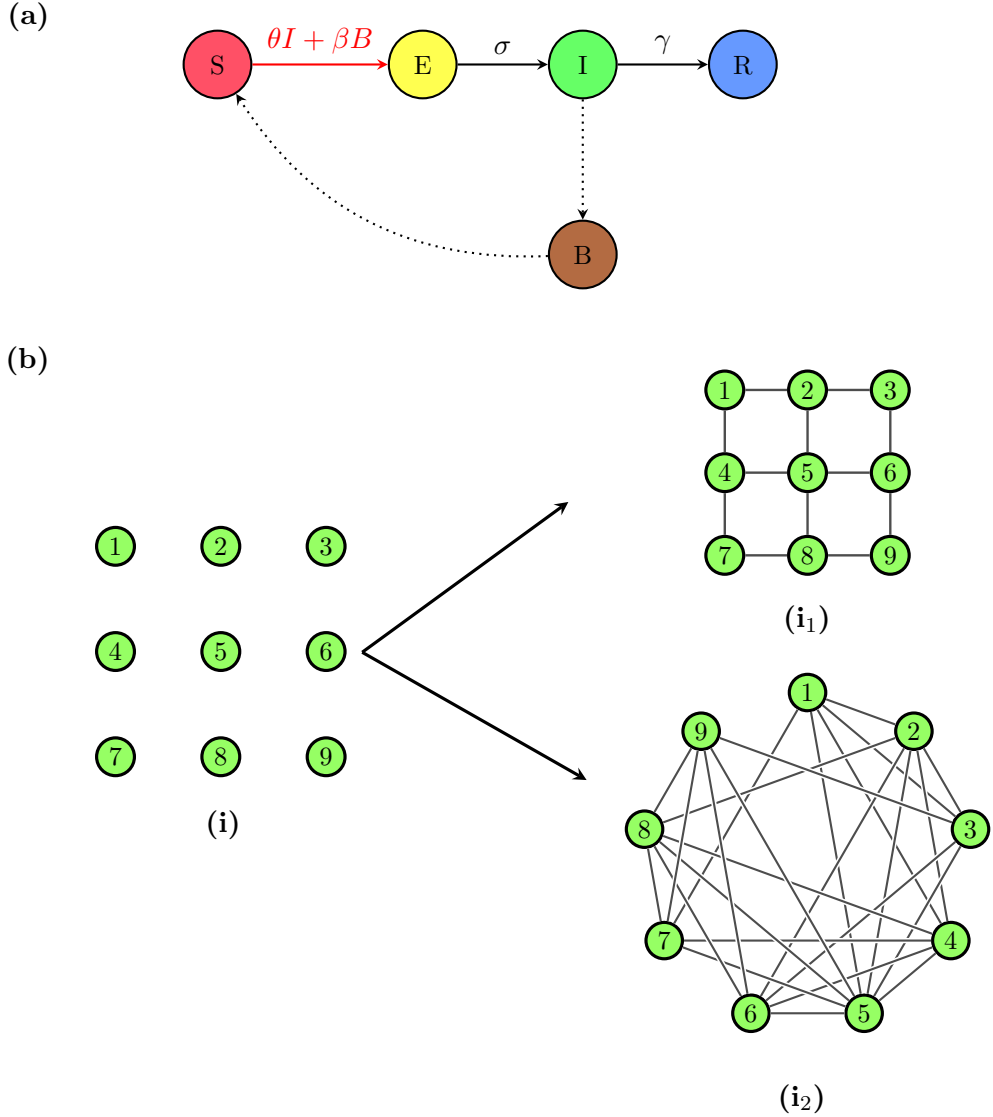
\begin{figure}[H]
\centering

\begin{subfigure}[t]{\textwidth}

\noindent\hspace*{1.8cm}\textbf{(a)}

\vspace{-0.2cm}

\noindent\hspace*{4cm}
\begin{tikzpicture}[
    node distance=0.3cm and 0.3cm,
    compartment/.style={circle, draw=black, minimum size=0.9cm, font=\small, thick, align=center},
    every path/.style={thick, ->, >=stealth}
]

\node[compartment, fill={rgb,255:red,255;green,80;blue,102}] (S) {S};
\node[compartment, fill={rgb,255:red,255;green,255;blue,80}, right=1.8cm of S] (E) {E};
\node[compartment, fill={rgb,255:red,102;green,255;blue,102}, right=1.2cm of E] (I) {I};
\node[compartment, fill={rgb,255:red,102;green,153;blue,255}, right=1.2cm of I] (R) {R};
\node[compartment, fill={rgb,255:red,179;green,107;blue,66}, below=1.6cm of I] (B) {B};

\path (S) edge[above, red] node {$\theta I + \beta B$} (E);
\path (E) edge[above] node {$\sigma$} (I);
\path (I) edge[above] node {$\gamma$} (R);
\path (I) edge[left, dotted] (B);
\path (B) edge[bend left, below, dotted] (S);

\end{tikzpicture}

\end{subfigure}

\vspace{0.6cm}

\begin{subfigure}[t]{\textwidth}
\centering

\begin{tikzpicture}[
  plant/.style={circle,draw=black,very thick,fill={rgb,255:red,150; green,255; blue,100},inner sep=2.2pt},
  edge/.style={line width=0.9pt,draw=black!70},
  arr/.style={->,very thick,>=stealth},
  box/.style={inner sep=0pt}
]
\node[anchor=north west,font=\bfseries] at (-3,3.6) {(b)};

\node[box] (F) at (0,-1) {%
\begin{tikzpicture}[scale=1.4]
\node[anchor=north west,font=\bfseries] at (0.85,-0.5) {(i)};
\foreach \i/\x/\y in {1/0/2, 2/1/2, 3/2/2,
                      4/0/1, 5/1/1, 6/2/1,
                      7/0/0, 8/1/0, 9/2/0}{
  \node[plant] (f\i) at (\x,\y) {\small \i};
}

\end{tikzpicture}
};

\node[box] (H) at (7.8,1.4) {%
\begin{tikzpicture}[scale=1.1]
\node[anchor=north west,font=\bfseries] at (0.7,-0.6) {($\text{i}_1$)};

\foreach \i/\x/\y in {1/0/2, 2/1/2, 3/2/2,
                      4/0/1, 5/1/1, 6/2/1,
                      7/0/0, 8/1/0, 9/2/0}{
  \node[plant] (f\i) at (\x,\y) {\small \i};
}
\draw[edge] (f1)--(f2); \draw[edge] (f2)--(f3);
\draw[edge] (f4)--(f5); \draw[edge] (f5)--(f6);
\draw[edge] (f7)--(f8); \draw[edge] (f8)--(f9);
\draw[edge] (f1)--(f4); \draw[edge] (f4)--(f7);
\draw[edge] (f2)--(f5); \draw[edge] (f5)--(f8);
\draw[edge] (f3)--(f6); \draw[edge] (f6)--(f9);
\end{tikzpicture}
};

\node[box] (C) at (7.8,-3.8) {%
\begin{tikzpicture}[
  scale=1,
  plant/.style={circle,draw=black,very thick,
    fill={rgb,255:red,150; green,255; blue,100},inner sep=2.0pt},
  edge/.style={
    line width=0.8pt,draw=black!70,
    preaction={draw=white,line width=2pt}
  }
]
\node[anchor=north west,font=\bfseries] at (-0.4,-3) {($\text{i}_2$)};
\foreach \i/\ang in {1/90,2/50,3/10,4/330,5/290,6/250,7/210,8/170,9/130}{
  \node[plant] (b\i) at (\ang:2.2) {\small \i};
}
\draw[edge] (b1)--(b2); \draw[edge] (b1)--(b3); \draw[edge] (b1)--(b4);
\draw[edge] (b1)--(b5); \draw[edge] (b1)--(b7);
\draw[edge] (b2)--(b3); \draw[edge] (b2)--(b4); \draw[edge] (b2)--(b5);
\draw[edge] (b2)--(b6); \draw[edge] (b2)--(b8);
\draw[edge] (b3)--(b5); \draw[edge] (b3)--(b6); \draw[edge] (b3)--(b9);
\draw[edge] (b4)--(b5); \draw[edge] (b4)--(b7); \draw[edge] (b4)--(b6); \draw[edge] (b4)--(b8);
\draw[edge] (b5)--(b6); \draw[edge] (b5)--(b7); \draw[edge] (b5)--(b8); \draw[edge] (b5)--(b9);
\draw[edge] (b6)--(b8); \draw[edge] (b6)--(b9);
\draw[edge] (b7)--(b8); \draw[edge] (b7)--(b9);
\draw[edge] (b8)--(b9);
\end{tikzpicture}
};

\draw[arr] ($(F.east)+(0.2,0.4)$) -- ($(H.west)+(-1.8,0)$);
\draw[arr] ($(F.east)+(0.2,0.4)$) -- ($(C.west)+(-0.5,1.5)$);

\end{tikzpicture}

\end{subfigure}

\caption{Illustration of the SEIRB modeling framework and network representations.
\textbf{(a)} Diagram of the SEIRB transmission structure, including direct and soil-mediated infection pathways.
\textbf{(b)} Effect of the distance threshold \(d\) on the plant contact network: \textbf{(i)} hypothetical soybean field layout; \textbf{($\text{i}_1$)} the network structure induced by a distance threshold of $d=1$, corresponds to nearest-neighbor connectivity on the lattice; and \textbf{($\text{i}_2$)} the network structure induced by a distance threshold of $d=2$, resulting in a denser pattern of connectivity.}
\label{fig:seirb_and_network}
\end{figure}

\subsection{Proposed Network-Based SEIRB Model}\label{sec:Network_model}

In network-based epidemic models, specify which individuals can transmit to one another, allowing transmission to be constrained by spatial proximity or biological contact structure~\cite{KeelingEames2005,Danon2011}. This distinction is essential for plant disease systems, where infections typically spread only between nearby hosts and where local clustering can substantially influence epidemic outcomes~\cite{ShawPautasso2014,Gilligan2008}.

To illustrate the difference between homogeneous mixing and spatially restricted contacts, 
consider a schematic graph $G=(V,E)$ with nine plants labeled 
$V=\{1,2,\ldots,9\}$, as shown in Fig.~\ref{fig:seirb_and_network}b(i). 
Under homogeneous mixing, each plant can transmit to every other plant, 
corresponding to the complete graph with edge set $
E_1=\{(i,j): 1 \le i < j \le 9\},$
so that $|E_1|=\binom{9}{2}=36$. 
In contrast, under the network model, each plant can transmit only to its immediate neighbors through a sparser contact graph. 
For example, the contact structure shown in Fig.~\ref{fig:seirb_and_network}b($\text{i}_1$) corresponds to the edge set $
E_2=\{(1,2),(1,4),(2,3),(2,5),(3,6),(4,5),(4,7),
(5,6),(5,8),(6,9),(7,8),(8,9)\},$
with $|E_2|=12$.

\subsubsection{Real Field network based on distance threshold}\label{sec:Real_field_distance}

We unravel the contact network structure of the soybean field based on the real field layout described in~\cite{Mengistu2014FLS}. In particular, the field consists of six subplots arranged as three till plots (top row) and three non-till plots (bottom row). The field layout follows a randomized split-split plot design in which each subplot has four rows of length \(6\) m, and rows are spaced at \(76.2\) cm. Soybeans are planted at \(12\) seeds per meter, giving \(72\) plants per row and a total of \(N=1728\) plants (Fig.~\ref{fig:field_layout}). 

We assign each plant a fixed coordinate \(\boldsymbol{x}_i\in\mathbb{R}^2\) based on its row and within-row position so that the geometry of every subplot, as well as the alleys between subplots, is preserved. Each plant \(i = 1,\dots,N\) is assigned to a subplot \(p_k(i) \text{ for } k \in \{1,\dots,6\}\) and to a management practice \(P(i) \in \{\text{till}, \text{non-till}\}\), with three subplots under tillage and three under non-till management, as shown in Fig.~\ref{fig:field_layout}. These plot and practice labels are used in the SEIRB model to define plot-specific soil reservoirs and practice-specific transmission and decay parameters, as well as management interventions such as plant removal.

Using the plant coordinates, we compute Euclidean distances
\(
d(i,j) = \lVert \boldsymbol{x}_i - \boldsymbol{x}_j \rVert_2
\)
for all pairs of plants. Plant-to-plant contacts are represented by an undirected proximity graph \(G = (V,E)\), where the vertex set \(V\) corresponds to individual plants and an edge \((i,j) \in E\) is present if \(0 < d(i,j) \le d\). The resulting contact structure is represented by an adjacency matrix \(A = (A_{ij})\), where
\[
A_{ij} =
\begin{cases}
1, & \text{if } 0 < d(i,j) \le d, \\
0, & \text{otherwise},
\end{cases}
\qquad
\mathcal{N}(i) = \{ j : A_{ij} = 1 \}
\]
and \(\mathcal{N}(i)\) denotes the set of neighboring plants that can directly transmit infection to plant \(i\).

Because the coordinates reflect the true physical spacing, edges across subplot boundaries appear only when \(d\) is large enough to bridge the alleys. The single parameter \(d\) allows us to tune the network from very local to more strongly connected regimes. The effect of the distance threshold on connectivity is illustrated for a hypothetical field of nine plants (Fig.~\ref{fig:seirb_and_network}b). As \(d\) increases, the network becomes progressively denser as additional diagonal and longer-range connections are added among neighboring plants. For demonstration purposes, Fig.~\ref{fig:seirb_and_network}b illustrates a hypothetical field represented as a graph, where plants are positioned on a $3\times3$ grid with unit spacing. Panel~($\text{i}$) depicts the base layout, while panel~($\text{i}_1$) corresponds to \(d = 1\), where only horizontal and vertical nearest neighbors are connected, forming a four-neighbor grid. For larger values of \(d\) (e.g., \(d = 2\); panel~($\text{i}_2$), displayed in a circular layout for readability), additional longer-range links emerge, producing a much denser network.

\begin{figure}[H]

  \centering
  \vspace{-0.1in}
  \includegraphics[width=0.88\linewidth]{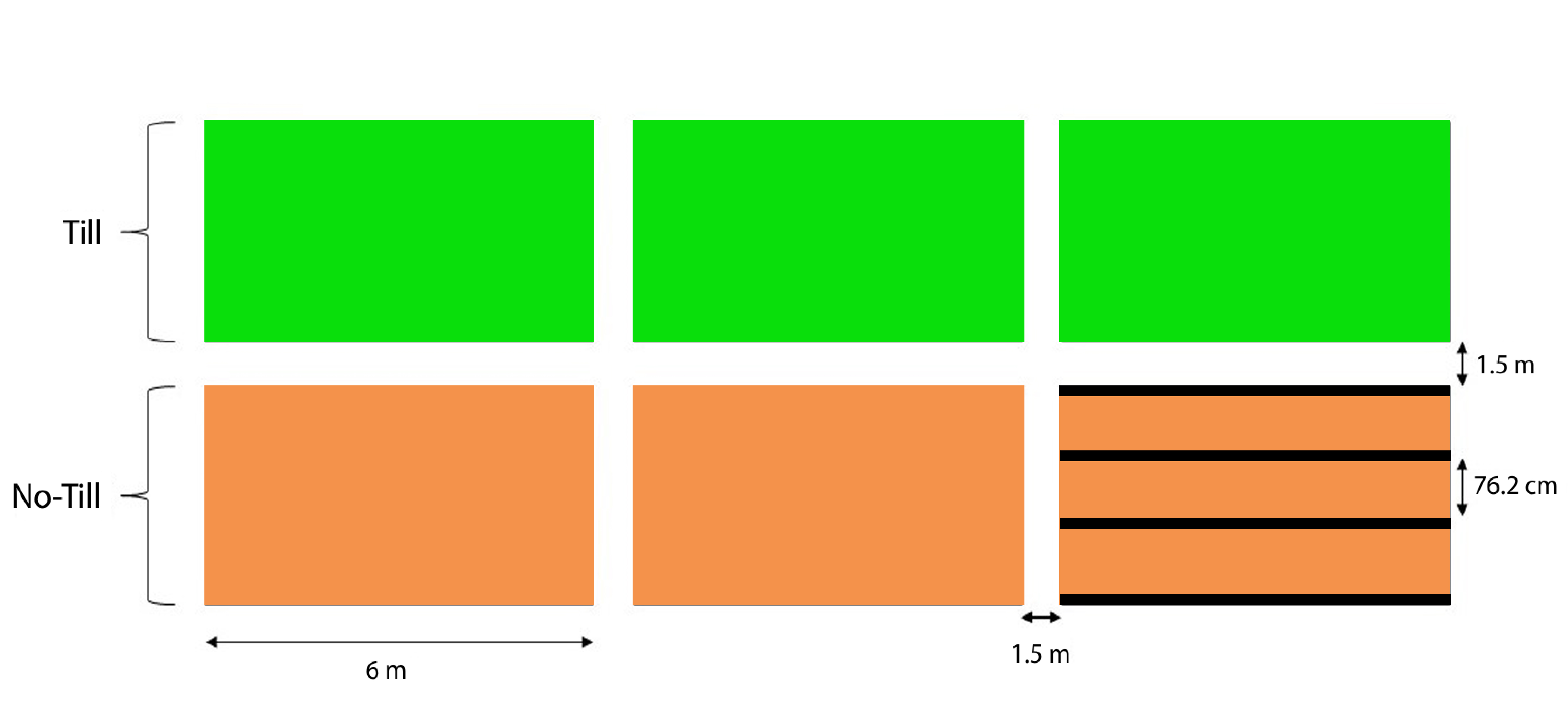}
  \vspace{0.1in}
  \begingroup
  \centering
  \begin{minipage}{0.85\textwidth}
  \justifying
  \caption[Real field layout used to construct the spatial network]{Real field layout used to construct the spatial network~\cite{Mengistu2014FLS}.}
 \label{fig:field_layout}
  \end{minipage}
   
  \par
  \endgroup
 
\end{figure}

\subsubsection{Proposed SEIRB Network Model}

Building upon the compartmental model in section~\ref{sec:priliminaries} and network structure as in section~\ref{sec:Real_field_distance}, this section presents our proposed SEIRB network model for FLS. The model couples distance-limited plant-to-plant transmission, defined through a distance-threshold contact graph derived from measured field geometry, with a plot-level environmental reservoir representing soil-borne inoculum. Management practices enter the model through practice-specific parameters governing exposure to and decay of environmental inoculum. We now describe the mathematical formulation of the spatial SEIRB network model for FLS. We consider a compartmental model with mutually exclusive disease states.  At each time \(t\), plant \(i\) occupies exactly one of the four epidemiological compartments, corresponding to the susceptible \((S_i)\), exposed \((E_i)\), infectious \((I_i)\), or removed \((R_i)\) states. This is represented by binary state variables

\[
S_i(t),\, E_i(t),\, I_i(t),\, R_i(t) \in \{0,1\},
\] which satisfy the constraint \[
S_i(t) + E_i(t) + I_i(t) + R_i(t) = 1 \qquad \forall i, i \in \{1,2,..,N\} \text{ and } t\in \{1,2,...,T\}.
\]

Each subplot \(p(i)\) contains a nonnegative soil reservoir \(B_p(t) \ge 0\), which accumulates inoculum from infectious plants and serves as a local source of new infections. The instantaneous infection hazard acting on plant \(i\) at time \(t\) is given by
\[
\lambda_i(t)
= \theta \sum_{j \in \mathcal{N}(i)} A_{ij}\, I_j(t)
+ \beta_{P(i)}\, B_{p(i)}(t),
\]
where \(\theta\) denotes the per-contact plant-to-plant transmission rate, the summation accounts for infectious neighboring plants \(j \in \mathcal{N}(i)\) connected to plant \(i\) in the contact network, and \(A_{ij}\) is the corresponding adjacency matrix entry. The second term represents soil-mediated transmission, where \(B_{p(i)}(t)\) denotes the level of infectious inoculum in the soil reservoir associated with subplot \(p(i)\), and \(\beta_{P(i)}\) is a management-specific soil-to-plant transmission coefficient that depends on whether plant \(i\) is grown under till or non-till conditions. 

Management practices modify the soil-mediated transmission pathway through tillage-specific multiplicative factors. Accordingly, the soil-to-plant transmission coefficient for plant $i$ is defined as
\[
\beta_{P(i)} =
\begin{cases}
\beta_{\text{till}}, & \text{if } P(i) = \text{till}, \\
\beta_{\text{non}},  & \text{if } P(i) = \text{non-till}.
\end{cases}
\]
Under tillage management, the soil-to-plant transmission coefficient and soil inoculum decay rate are modeled as
\[
\beta_{\text{till}} = \rho_{\beta}\,\beta_{\text{non}},
\qquad
\tau_{\text{till}} = \rho_{\tau}\,\tau_{\text{non}},
\]
where $\beta_{\text{non}}$ and $\tau_{\text{non}}$ denote the baseline soil exposure and decay parameters under non-till management, and $\rho_{\beta}$ and $\rho_{\tau}$ are dimensionless multipliers that quantify the proportional effect of tillage on transmission and decay, respectively.

The soil inoculum reservoir \(B_{p}(t)\) associated with subplot \(p\) evolves according to
\[
\dot B_{p}(t)
= r_{0}\,B_{p}(t)\!\left(1-\frac{B_{p}(t)}{k_{0}}\right)
- \tau_{P(p)}\,B_{p}(t)
+ \xi \sum_{j:\,p(j)=p} I_j(t),
\]
where the first term represents logistic growth of soil-borne inoculum with intrinsic growth rate \(r\) and saturation at the carrying capacity \(k\), the second term represents practice-specific decay or removal of inoculum at rate \(\tau_{P(p)}\), and the third term represents shedding of inoculum from infectious plants located in subplot \(p\). Here, \(\xi\) is the per-infectious-plant shedding rate, and the summation is taken over all plants \(j\) assigned to subplot \(p\).

The discrete-time Euler update with one-day steps is
\[
B_p(t{+}1)
= B_p(t)
+ r_{0}\,B_p(t)\!\left(1-\frac{B_p(t)}{k_{0}}\right)
- \tau_{P(p)}\,B_p(t)
+ \xi\,I_p(t),
\]
with nonnegativity enforced by \(B_p(t{+}1)=\max\{0,B_p(t{+}1)\}\). 

From time \(t\) to \(t{+}1\), a susceptible plant \(i\) becomes exposed with probability,
\[
p_i(t{+}1) = 1 - \exp\!\bigl[-\lambda_i(t)\bigr],
\]
where \(\lambda_i(t)\) is the instantaneous infection hazard at time \(t\).
Progression from the exposed to the infectious state and removal from the infectious state occur with constant daily transition probabilities,
\[
p_{E\to I} = 1 - \exp(-\sigma_{0}),
\qquad
p_{I\to R} = 1 - \exp(-\gamma_{0}),
\]
which are applied uniformly to all exposed and infectious plants, respectively.
All transition probabilities are evaluated using the system state at time \(t\), and state updates are performed synchronously to obtain the configuration at time \(t{+}1\).

All fixed biological parameters $(\sigma_{0}, \gamma_{0}, r_{0}, k_{0}, B(0))$ were adopted from~\cite{MBE2024FLS}, and the initial soil inoculum level was set identically for both till and non-till subplots, with $B_p(0) = B(0)$.
The remaining transmission and management parameters are treated as unknown and estimated from the observed epidemic data, and are collected in the parameter vector
\(
\boldsymbol{\vartheta}
= \big(\theta,\ \beta_{\mathrm{non}},\ \rho_{\beta},\ \xi,\ \tau_{\mathrm{non}},\ \rho_{\tau},\ d\big).
\)

\section{Results} \label{sec:results}
The Results section presents the main findings of this study. In particular, Section~\ref{sec:Res_data} evaluates the fit of the proposed model and examines the influence of initial infection geometry. Section~\ref{sec:res_tillage} studies the effects of tillage practices using parameter multipliers and cumulative disease burden. Finally, Section~\ref{sec:res_rogue} assesses the effectiveness of roguing intervention strategies on the real-field network.

\subsection{Data Analysis}\label{sec:Res_data}

\subsubsection{Parameter Estimation via approximate Bayesian computation (ABC)}\label{sec:res_parameters}

Approximate Bayesian computation (ABC) provides a likelihood-free framework for estimating the parameters of our stochastic, spatially explicit SEIRB model. Because the model generates path-dependent epidemic trajectories, the likelihood function is analytically intractable. ABC overcomes this difficulty by repeatedly drawing candidate parameter values from the prior, simulating epidemic outcomes, and retaining only those parameters whose simulated trajectories closely match the observed data according to a chosen discrepancy measure and tolerance threshold. As the tolerance is tightened, the accepted parameters increasingly approximate draws from the true posterior distribution. For reference, the basic rejection-based ABC procedure is summarized as in the pseudocode in Algorithm below. Many variants of ABC have been developed to improve convergence performance \cite{Franks01102020,Pudlo2016,Raynal2019}. In this study, we adopt the replenishment ABC (RABC) method of Drovandi and Pettitt~\cite{DrovandiPettitt2011_SMC_ABC}.

\renewcommand{\thealgorithm}{}
\begin{algorithm}[H]
\caption*{\textbf{Pseudocode:} ABC Rejection Algorithm}
\begin{algorithmic}[1]
 \State Sample \(\theta^{(i)} \sim \pi(\theta)\) from the prior.
 \State Simulate data \(y^{(i)} \sim f(y \mid \theta^{(i)})\).
 \State Compute the discrepancy \(d = d(S(y^{(i)}), S(y_{\mathrm{obs}}))\).
 \If{$d \le \varepsilon$} \State accept \(\theta^{(i)}\)
 \Else \State reject
 \EndIf
\end{algorithmic}
\end{algorithm}

We calibrate the model using the empirical disease-severity dataset summarized in Table~\ref{tab:fls_severity}. Because the initial infection locations were not recorded, we randomly distributed the 52 infections across the mapped field layout (Fig.~\ref{fig:fields}a) and used this configuration as the default initialization for all model runs in this study. 

Priors were specified as:
$\theta \sim \mathrm{LogNormal}(-10.98,1.22)$,
$\beta_{\mathrm{non}} \sim \mathrm{LogNormal}(-18.31,1.97)$,
$\xi \sim \mathrm{LogNormal}(4.99,1.82)$,
$\tau_{\mathrm{non}} \sim \mathrm{LogNormal}(-4.18,2.23)$,
$\rho_{\beta}\sim \mathrm{U}(0.1,2.0), \rho_{\tau} \sim \mathrm{U}(0.1,2.0)$, and
$d \sim \mathrm{U}(70,310)$.
These priors were derived from a data-driven pilot search to improve convergence, based on preliminary model runs while retaining sufficient variability for the spatial network model. Further implementation details are available in the accompanying GitHub repository. The distance metric used is the absolute-error distance
$d(I_{\mathrm{sim}}, I_{\mathrm{obs}}) = \sum_{t \in T} \lvert I_{\mathrm{sim}}(t) - I_{\mathrm{obs}}(t) \rvert$,
where $T$ denotes the set of observation times, $I_{\mathrm{obs}}(t)$ is the number of infected plants at time $t$, and $I_{\mathrm{sim}}(t)$ is the corresponding simulated count. Using these priors and the distance metric within RABC, we obtain the 100 best posterior parameter combinations and use them to represent the model parameter estimates.

To assess model fit, we simulated three forward epidemic trajectories for each of 100 posterior parameter samples, yielding 300 trajectories. We retained the 100 with the smallest distance to the observed data and computed the predictive mean and 95\% credible intervals from this subset. As shown in Fig.~\ref{fig:traj_fit_7par}, the predictive mean closely tracks the observed infection counts throughout the season, indicating strong agreement with the field data. The 95\% credible intervals encompass the observed infection counts for most of the season. Together, these results indicate good agreement between the model and the field data.

\begin{table}[ht]
\centering
\caption{FLS disease severity data~\cite{MBE2024FLS}.}
\label{tab:fls_severity}
\small
\begin{tabular}{lcccccccc}
\toprule
\textbf{Days after planting} & 0 & 45 & 50 & 75 & 89 & 96 & 117 & 138 \\
\midrule
\textbf{Disease severity}    & 3\% & 5\% & 6\% & 8\% & 16\% & 21\% & 28\% & 36\% \\
\midrule
\textbf{Infected plants (count)} & 52 & 86 & 104 & 138 & 276 & 362 & 484 & 622 \\
\bottomrule
\end{tabular}
\end{table}
\vspace{0.2in}

\begin{figure}[ht]
  \centering
  \includegraphics[width=0.7\textwidth]{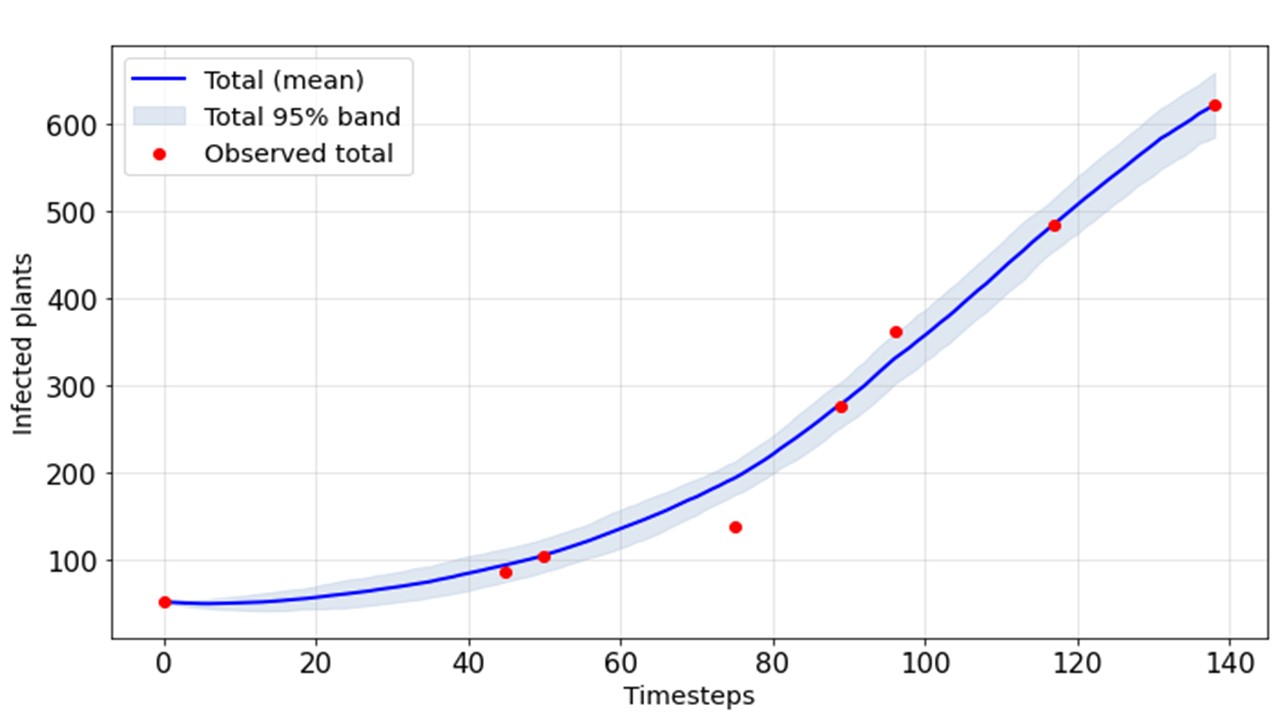} 
\begingroup
\centering 
\begin{minipage}{0.85\textwidth} 
\justifying

\caption{Observed infected-plant counts and posterior predictive trajectories. Solid line: mean of the best-100 trajectories; the shaded ribbon denotes the 95\% credible band. The best-fit model parameter vector is $
\vartheta
=
\left(
\theta,\,
\beta_{\text{non}},\,
\rho_\beta,\,
\xi,\,
\tau_{\text{non}},\,
\rho_\tau,\,
d
\right)
=
\left(
3.90\times 10^{-5},\,
4.24\times 10^{-8},\,
1.45,\,
113.5,\,
0.263,\,
0.881,\,
83.1
\right).$}
\label{fig:traj_fit_7par}
\end{minipage}
\endgroup

\end{figure}
\vspace{-1cm}


\subsubsection{Sensitivity Analysis Across Initial Infection Scenarios}\label{sec:res_sensitivity_initial}

To evaluate the sensitivity of parameter estimates to initial infection geometry, we considered alternative seeding patterns. The random configuration with $|S_0|=52$ used for ABC calibration (Section~\ref{sec:res_parameters}) serves as the baseline for comparison with the clustered and polycentric seeding patterns. The clustered and polycentric arrangements are defined below.

\begin{enumerate}
    \item[{\textit{1.}}] \textit{Cluster}: A dense, compact focus of infection occupying a single region of the field. Operationally, one centroid is selected at a plant location, and $S_0$ is defined as the $52$ nearest plants (equivalently, all plants within a radius $r$ chosen such that $|S_0|=52$). This produces strong spatial autocorrelation and minimal fragmentation.

    \item[{\textit{2.}}] \textit{Polycentric}: Initial infections are distributed around $K\!\ge\!2$ well-separated centroids, subject to a minimum inter-centroid distance. The 52 seeds are partitioned across these regions (approximately balanced $n_k$ with $\sum_{k=1}^K n_k=52$), assigning each infected plant to the closest centroid and ensuring non-overlapping neighborhoods. 
    
\end{enumerate}

\begin{figure}[H]
\centering

\begin{subfigure}{\textwidth}
    \centering
    \includegraphics[width=\textwidth]{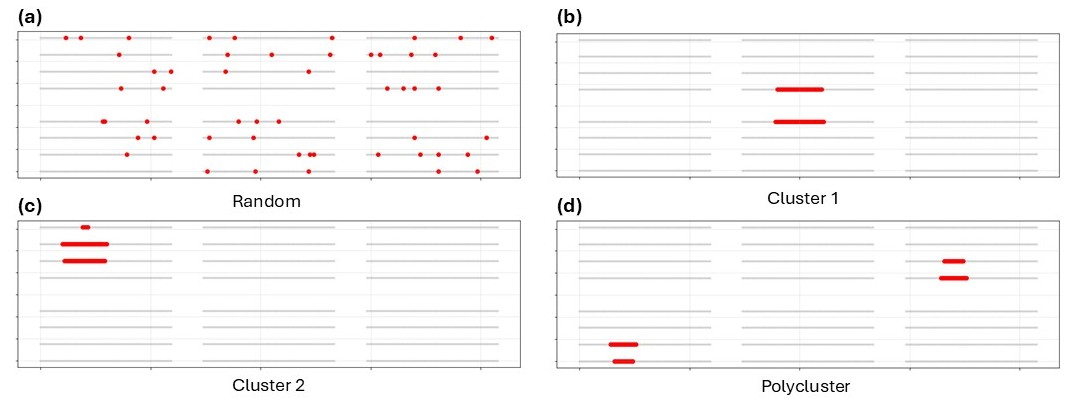}
\end{subfigure}

\begingroup
\centering
\begin{minipage}{0.98\textwidth}
\justifying

\caption{Spatial configurations of the initially infected plants for the four seeding scenarios: \textit{Random}, \textit{Cluster~1}, \textit{Cluster~2}, and \textit{Polycluster}.}
\label{fig:fields}
\end{minipage}
\par
\endgroup
\end{figure}

In our analysis, these constructions were instantiated as three practical scenarios. In \textit{Cluster~1} (Fig.~\ref{fig:fields}b), infections are assigned to the 52 plants nearest to a centroid located near the field center, forming a compact, interior outbreak that tests sensitivity to centrally initiated epidemics.  In \textit{Cluster~2}(Fig.~\ref{fig:fields}c), all 52 initial infections are placed on one side of the field, typically within a single subplot, generating a strongly asymmetric starting condition and localized early amplification.  In the \textit{Polycluster} scenario (Fig.~\ref{fig:fields}d), we adopt the bicentric case ($K=2$), positioning two non-overlapping clusters in distinct subplots (one till, one non-till) to assess cross-plot transmission and management effects. Each cluster is generated by selecting a centroid and assigning the nearest plants, with neighborhoods constrained to avoid overlap. These structured seeding patterns allow us to investigate how early spatial heterogeneity shapes epidemic timing, cross-plot spread, and the relative contributions of direct and soil-mediated pathways under a fixed parameter set.

The sensitivity of the calibrated model parameters to different initial-infection scenarios is illustrated by comparing the marginal posterior distributions of the seven parameters across the four spatial configurations (Fig.~\ref{fig:posterior_boxplots}). The most pronounced scenario-dependent differences occur in the spatial distance threshold $d$, with the Random scenario concentrating at smaller values, while the clustered and Polycluster scenarios favor substantially larger spatial scales. Soil-related parameters also vary across scenarios: the shedding rate $\xi$ tends to be lower in the Random scenario and higher in clustered configurations, whereas the baseline soil exposure parameter $\beta_{\mathrm{non}}$ shows the opposite pattern. In contrast, the plant-to-plant transmission rate $\theta$, the soil decay parameter $\tau_{\mathrm{non}}$, and the tillage multipliers $(\rho_{\beta}, \rho_{\tau})$ show comparatively smaller scenario-dependent differences, with substantial overlap in their posterior distributions. Overall, the figure shows that different initial infection geometries lead to systematic differences in the inferred values of spatial and soil-related parameters, even when all scenarios fit the same observed infection data. This finding highlights the importance of collecting precise initial infection locations in future field studies, as such data would enhance the robustness of model parameter estimation.

\begin{figure}[H]
\centering

\begin{subfigure}{\textwidth}
    \centering
    \includegraphics[width=\textwidth]{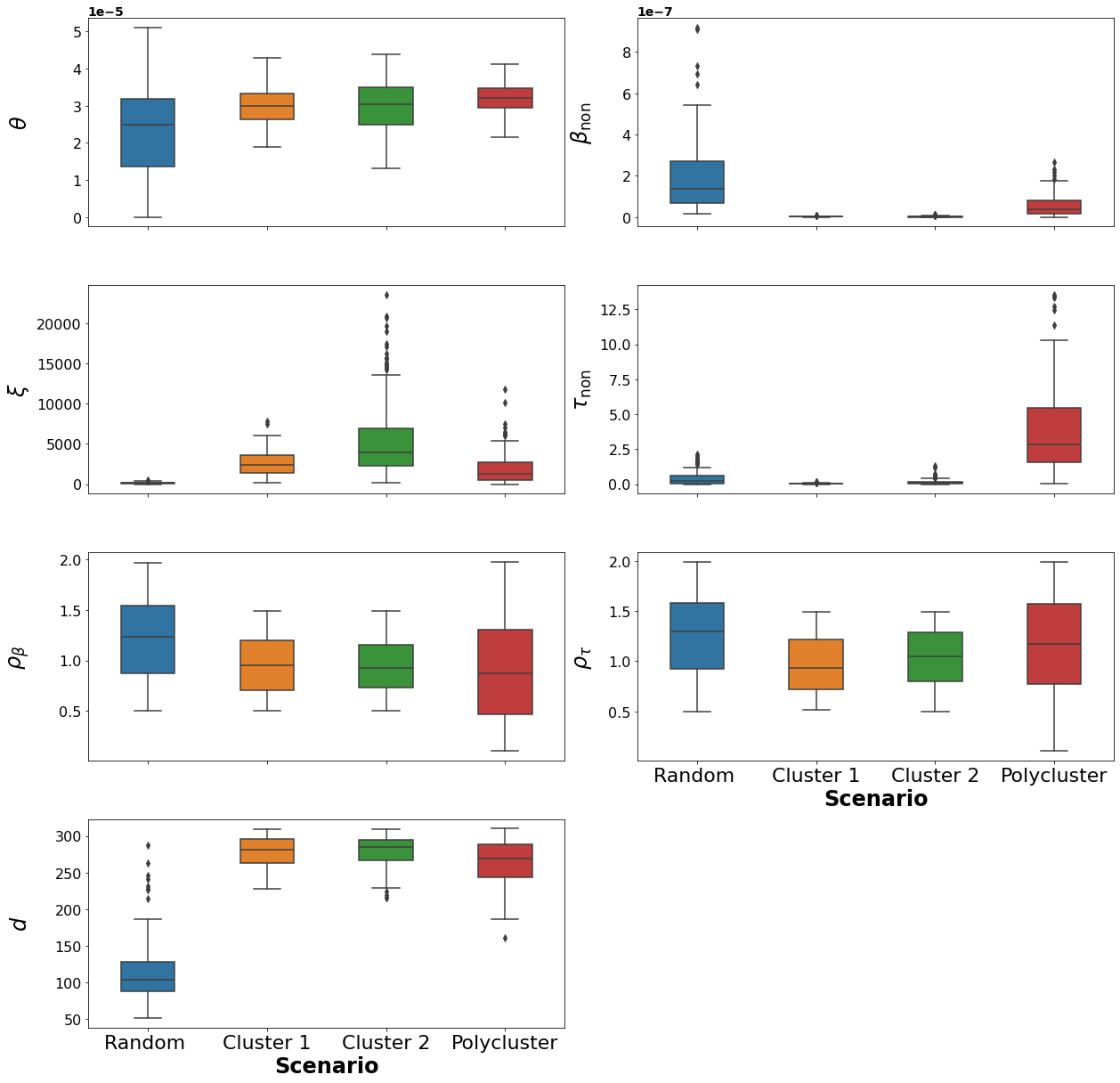}
\end{subfigure}

\begingroup
\centering
\begin{minipage}{0.98\textwidth}
\justifying

\caption{Posterior sensitivity to initial seeding, shown as boxplots of marginal posterior samples for $\vartheta=(\theta,\beta_{\mathrm{non}},\xi,\tau_{\mathrm{non}},\rho_\beta,\rho_\tau,d)$ under the four seeding scenarios.
}
\label{fig:posterior_boxplots}
\end{minipage}
\par
\endgroup
\end{figure}

\subsection{Assessing the Effectiveness of Tillage Practices}\label{sec:res_tillage}

\subsubsection{Tillage effects on fungus spread and decay}\label{sec:tillage_multipliers}

We evaluated whether tillage modifies the soil-mediated transmission pathway through two practice-specific multipliers: (i) an exposure multiplier \(\rho_{\beta}\), which tests whether tillage reduces soil-driven infection pressure, and (ii) a decay multiplier \(\rho_{\tau}\), which tests whether tillage accelerates residue decay.

We set the two hypotheses as follows,
\begin{equation}\label{eq:test1}
H_{01}:\rho_{\beta}=1 \text{ vs.} \ H_{11}:\rho_{\beta}<1, 
\end{equation}
\begin{equation}\label{eq:test2}
H_{02}:\rho_{\tau}=1 \text{ vs.} \ H_{12}:\rho_{\tau}>1, 
\end{equation}

For the testing problem \eqref{eq:test1}, if \(H_{01}\) is rejected in favor of \(H_{11}\), then there is evidence that tillage lowers soil exposure (\(\rho_{\beta}<1\)).

For the testing problem \eqref{eq:test2}, rejecting \(H_{02}\) in favor of \(H_{12}\) provides evidence that tillage raises the decay rate (\(\rho_{\tau}>1\)).

Table \ref{tab:onesided_ci_random} represents the 95\% CI and $p$-value for the above two tests. Based on this result, we fail to reject $H_{01}$ and $H_{02}$, and there is no statistically credible evidence that tillage reduces soil-driven exposure or sufficiently speeds decay under the current data and calibration.

\begin{table}[ht]

\centering
\caption{One-sided credible intervals and $p$-values for the exposure and decay ratios.}
\label{tab:onesided_ci_random}

\begin{tabular}{l l  c  c}
\toprule
Quantity & Alternative &  CI$_{0.95}$ &  $p$-value \\
\midrule
$\rho_\beta$ (exposure ratio) & $H_{11}:\rho_\beta<1$ &  $(-\infty,\,1.8756)$ &  0.645 \\
$\rho_\tau$ (decay ratio)     & $H_{12}:\rho_\tau>1$ &  $(0.5638,\,\infty)$ &  0.315 \\
\bottomrule
\end{tabular}

\end{table}

\begin{figure}[ht]
  \centering
  \includegraphics[width=0.8\textwidth]{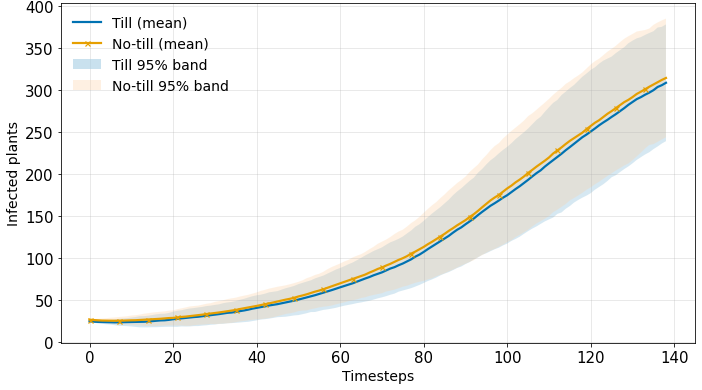}

  \caption{Till vs.\ non-till infection trajectories}
  
  \label{fig:till_notill_infected}
\end{figure}

\subsubsection{Tillage effects on disease severity
}\label{sec:AUDPC}

To further understand how tillage practices influence disease severity over the course of the growing season, we quantify epidemic intensity using the \textit{Area Under the Disease Progress Curve (AUDPC)}. The AUDPC is a standard metric to measure disease severity in plant pathology~\cite{Mengistu2014FLS}. At each time step $t$, the number of infected plants in a given management practice (till or non-till) is recorded. Dividing this value by the total number of plants in that practice yields the disease severity, which represents the proportion of plants infected at time $t$, ranging between 0 (no infection) and 1 (all plants infected). The AUDPC is then obtained by integrating the severity curve across the epidemic timeline using the trapezoidal rule:
\[
\mathrm{AUDPC} = \sum_{i=1}^{T-1} \frac{\mathrm{Severity}(t_i) + \mathrm{Severity}(t_{i+1})}{2} \cdot (t_{i+1} - t_i),
\]
where $T$ is the total number of time steps.

Each simulation infection count over time was recorded for both till $I_{\mathrm{till}}(t)$ and non-till $I_{\mathrm{non\mbox{-}till}}(t)$ plots. For each practice, counts were converted to severity (fraction infected) by dividing by the number of plants in that practice, and the trapezoidal rule was applied across all 139 times teps to compute the AUDPC.

To evaluate whether one practice tends to produce greater cumulative severity, we used the one-sided binomial test. Let us denote $p=\Pr\!\big(\mathrm{AUDPC}_{\mathrm{till}}<\mathrm{AUDPC}_{\mathrm{non-till}}\big)$. The hypotheses were specified as
\begin{equation}\label{eq:test3}
    H_{03}: p = 0.5  \text{ vs. }  H_{13}: p > 0.5
\end{equation}

For the test \eqref{eq:test3}, we cannot reject $H_{03}$ since the one-sided $95\%$ confidence interval for the parameter is $[0.404,1.000]$ that includes the value $0.5$ within the interval with the $p$-value $=0.6178$. Also, Figure \ref{fig:till_notill_infected} confirms that there is no significant deviation between till and non-till practices. Hence, we can conclude that, under the current model and parameter draws, there is no statistically significant difference in cumulative disease burden between till and non-till. Therefore, tillage showed no statistically significant reduction in disease burden under the tested conditions.

\subsection{Roguing Intervention}\label{sec:res_rogue}

This section examines roguing as a practical disease-management strategy for FLS, focusing on how the timing, frequency, and targeting of plant removal influence epidemic outcomes across the entire soybean field. Soybean moves from vegetative (V) to reproductive (R) stages; as the canopy closes and humidity rises, foliar pathogens intensify. An \emph{early} window near day~35 offers leverage because plant-to-plant contacts are fewer and the soil inoculum reservoir is still small~\cite{Nleya2019SoybeanStages}. Progressively \emph{later} windows face a denser contact network and more accumulated inoculum, so removals give diminishing returns even if more plants are taken out. We evaluated \emph{late roguing} starting at 42~days after planting (DAP), a stage when rows remain accessible and infections are still localized. Roguing at this time can eliminate nascent foci before widespread flowering.

Using the proposed model, we evaluated plant removal (roguing) as a disease management intervention. All simulations used the same parameter set and the same fixed initial condition $S_0$ with $|S_0|=52$ infected plants, ensuring that differences among scenarios arise solely from the timing, strategy, and frequency of removal. Two strategies were considered: random removal of infected plants and targeted removal of infected plants with the highest degree in the proximity graph. Both approaches were implemented at equal total effort, with $35$ ($2\%$ total removal) plants removed per season. Each strategy was examined under two start times, an early intervention at day~35 and a late intervention at day~42, and under three removal frequencies, with events occurring every $\Delta t=\{1,3,7\}$ days.

Roguing reduces transmission through both direct and soil-mediated pathways by lowering the number of infectious plants and limiting future shedding into the environmental reservoir. As a result, the timing of roguing strongly influences its impact. In all configurations, early interventions consistently produced smaller epidemic peaks and more healthy plants by harvest than interventions beginning only a week later. Higher-frequency removal amplified these benefits: daily roguing ($\Delta t=1$) generated the greatest overall reduction in epidemic intensity, while weekly schedules still produced measurable improvements relative to the no-roguing baseline. Strategy also played a decisive role. Targeted removal was uniformly more effective than random removal, reflecting the advantage of eliminating highly connected plants that act as transmission hubs within the proximity network.

Figure~\ref{fig:rogue_six} illustrates the epidemic trajectories for each combination of timing, frequency, and strategy, and Table~\ref{tab:rogue_core} reports the corresponding peak sizes, peak timings, and end of season healthy plant counts. The strongest performing intervention, which is early daily targeted roguing, resulted in $757\pm57$ healthy plants at day~138, a substantial improvement over the $655\pm38$ healthy plants produced by the late targeted strategy and dramatically higher than the $319\pm24$ obtained with no roguing. In relative terms, early targeted roguing produced approximately 15.6\% more healthy plants than its late counterpart, 26\% more than late random roguing, and more than double the number produced in the absence of any removal. In contrast, the weakest intervention, which is late weekly random roguing, offered only modest gains, with $601\pm48$ healthy plants only slightly above the no roguing outcome. Across all scenarios, the timing of the epidemic peak remained largely unchanged, indicating that roguing primarily reduces peak magnitude rather than delaying its occurrence.

\begin{table}[H]
\caption{Mean $\pm$ SD epidemic outcomes by roguing frequency ($\Delta t = 1, 3, 7$ days), intervention start day (early = 35, late = 42), and roguing strategy (none, random, targeted).}
\label{tab:rogue_core}
\resizebox{\textwidth}{!}{%
\begin{tabular}{c c l c c c}
\toprule
$\boldsymbol{\Delta t}$ & \textbf{Start} & \textbf{Roguing Strategy} & \textbf{Peak infected (plants)} &
\textbf{Peak Time(days)} &  \textbf{Healthy} \\
\midrule
\multirow{6}{*}{1}
 & 35 & No Roguing         & $709 \pm 29$ & $73 \pm 5$ &  $319 \pm 24$ \\
 & 35 & Random Roguing     & $543 \pm 25$ & $50 \pm 5$ &  $637 \pm 34$ \\
 & 35 & Targeted Roguing   & $519 \pm 39$ & $50 \pm 3$  & $757 \pm 57$ \\
 & 42 & No Roguing         & $709 \pm 29$ & $73 \pm 5$ & $319 \pm 24$ \\
 & 42 & Random Roguing     & $608 \pm 23$ & $51 \pm 2$ & $584 \pm 33$ \\
 & 42 & Targeted Roguing   & $597 \pm 33$ & $53 \pm 3$ & $676 \pm 57$ \\
\midrule
\multirow{6}{*}{3}
 & 35 & No Roguing         & $709 \pm 29$ & $73 \pm 5$ &  $319 \pm 24$ \\
 & 35 & Random Roguing     & $541 \pm 31$ & $50 \pm 3$ &  $659 \pm 49$ \\
 & 35 & Targeted Roguing   & $547 \pm 25$ & $50 \pm 4$ &  $710 \pm 47$ \\
 & 42 & No Roguing         & $709 \pm 29$ & $73 \pm 5$ & $319 \pm 24$ \\
 & 42 & Random Roguing     & $608 \pm 30$ & $53 \pm 3$ & $596 \pm 32$ \\
 & 42 & Targeted Roguing   & $597 \pm 41$ & $54 \pm 4$ & $671 \pm 57$ \\
\midrule
\multirow{6}{*}{7}
 & 35 & No Roguing         & $709 \pm 29$ & $73 \pm 5$ & $319 \pm 24$ \\
 & 35 & Random Roguing     & $545 \pm 39$ & $52 \pm 4$ & $652 \pm 67$ \\
 & 35 & Targeted Roguing   & $526 \pm 32$ & $52 \pm 5$ & $725 \pm 51$ \\
 & 42 & No Roguing         & $709 \pm 29$ & $73 \pm 5$ & $319 \pm 24$ \\
 & 42 & Random Roguing     & $606 \pm 31$ & $53 \pm 3$ & $601 \pm 48$ \\
 & 42 & Targeted Roguing   & $613 \pm 27$ & $54 \pm 3$ & $655 \pm 38$ \\
\bottomrule
\end{tabular}%
}

\end{table}

\begin{figure}[H]
  \centering
  \includegraphics[width=1\textwidth]{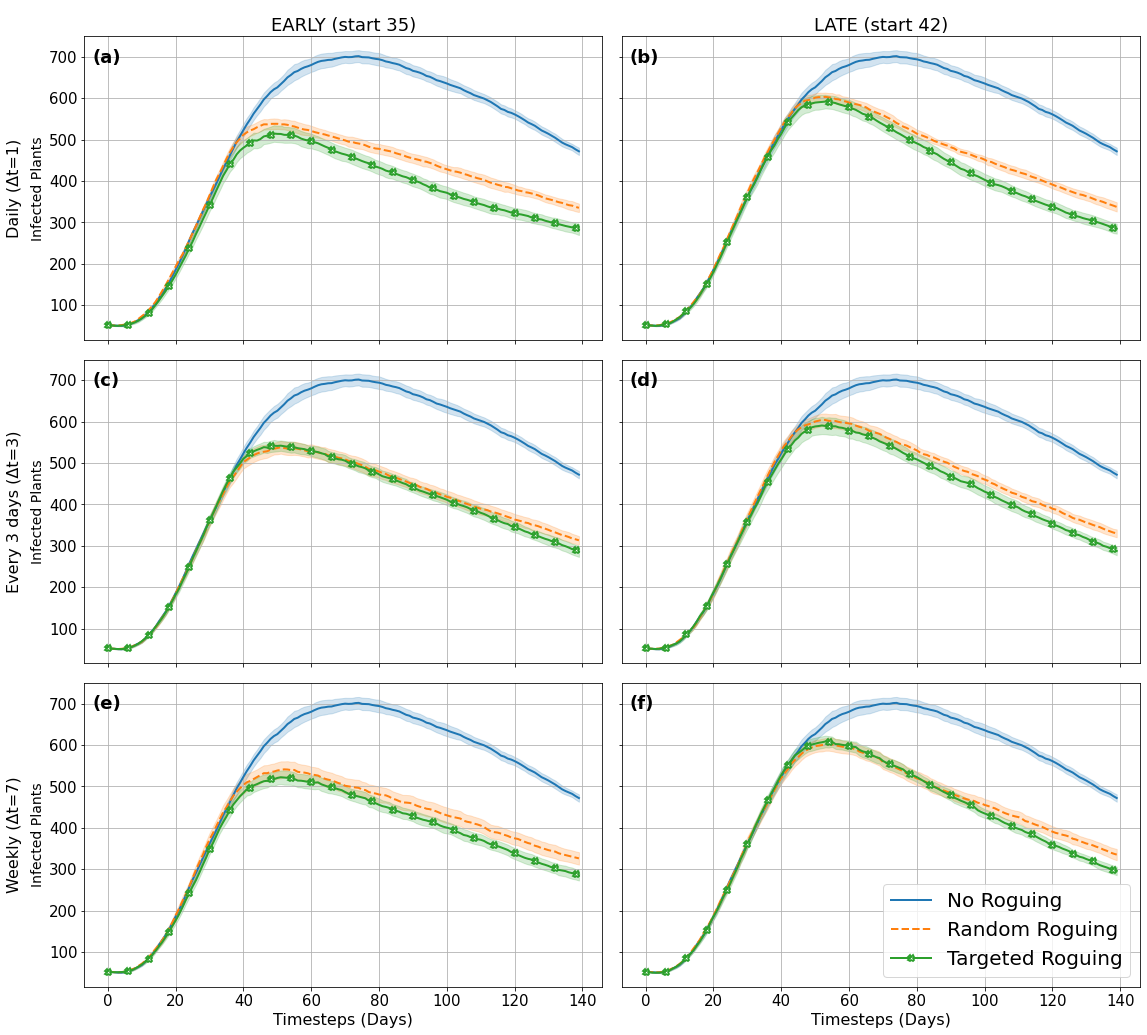}
\begingroup
\centering
\begin{minipage}{0.85\textwidth}
\justifying
 \caption{Timing (early vs.\ late ) and roguing interval effects on infections. Early (left column; start $=35$~d) vs.\ late (right column; start $=42$~d) roguing at three roguing frequency: daily ($\Delta t{=}1$; top), every 3~days ($\Delta t{=}3$; middle), and weekly ($\Delta t{=}7$; bottom). Curves represent replicate means with 95\% CIs; strategies include no roguing (blue, solid), random (orange, dashed), and targeted (green, dotted).}
 \label{fig:rogue_six}
\end{minipage}
\par
\endgroup
 
\end{figure}

Overall, the results reveal three consistent patterns: initiating removal early is markedly more effective than starting later; more frequent removal substantially strengthens control, particularly when targeting high-degree plants; and targeted roguing provides a persistent advantage over random removal at equal effort. These findings underscore that both the timing and network structure of roguing interventions play critical roles in determining their effectiveness at the field scale.

\section{Discussion and Conclusion}\label{sec:conclusion}
We proposed a spatially explicit network-based SEIRB framework to study Frogeye Leaf Spot (FLS) dynamics in soybean fields. The model represents plant-to-plant transmission through proximity-based network connections while simultaneously incorporating a plot-level soil reservoir that captures environmental persistence and pathogen reintroduction. Coupling these above-ground and soil-mediated processes within the measured field geometry provides a mechanistic description of disease spread beyond homogeneous-mixing assumptions. To estimate epidemiological parameters under observational uncertainty, we employed an approximate Bayesian computation (ABC) scheme, enabling likelihood-free inference of transmission and decay parameters. 
The integration of ABC with the spatial SEIRB network structure yields a transparent modeling framework that enables quantitative evaluation of management strategies, including roguing and tillage, and offers practical insight into disease propagation in heterogeneous agricultural systems.

The results demonstrate that the ABC-based calibration effectively estimated key model parameters and highlighted the strong influence of initial infection patterns on subsequent epidemic dynamics. Simulations further suggested that tillage practice, under the examined field conditions, did not produce a measurable difference in epidemic outcomes, implying that tillage alone may not substantially alter FLS transmission risk. In contrast, early and consistent roguing interventions produced substantial reductions in epidemic severity and final infection size, underscoring the importance of timely and targeted removal for effective disease control.

Early roguing consistently produced stronger epidemiological benefits, reducing both the peak prevalence and the final epidemic size across all simulated timing and intensity combinations. These improvements arise because, during late vegetative to early blooming stages, canopy closure is still limited and the soil inoculum reservoir remains comparatively low; removing infectious plants at this stage prevents substantial future shedding and restricts the growth of the environmental reservoir. In contrast, later roguing occurs under denser canopy contact and elevated soil inoculum, making infection hotspots more established and harder to eliminate even when larger fractions of plants are removed. The simulations also show that frequent scouting, ideally every day although weekly removal still provides measurable benefit, and prioritizing plants with many neighbors, such as those at row intersections, dense patches, or structurally connected bridge plants, leads to the greatest reductions in disease burden. Taken together, these results indicate that the most effective roguing strategy combines early intervention, frequent scouting, and targeted removal, yielding the lowest epidemic levels and the healthiest canopy at the end of the season.

While the proposed SEIRB framework captures the major processes driving FLS spread, several limitations should be acknowledged. First, the model currently assumes constant environmental conditions and does not explicitly incorporate weather variables such as temperature, humidity, and rainfall, which are known to influence fungal sporulation, dispersal, and decay. 
Incorporating weather-dependent transmission and decay functions would improve predictive accuracy and seasonal realism. Second, the spatial configuration of the initially infected plants was reconstructed from random or idealized seeding scenarios because their exact field locations were not recorded. This uncertainty may affect the quantitative calibration of spatial transmission parameters and the resulting epidemic trajectories.  Obtaining precise initial infection maps from future field observations would enable stronger validation and improve parameter identifiability. Despite these constraints, the model provides a strong baseline framework for integrating environmental and spatial heterogeneity in future extensions.

Although the model was formulated for FLS, the underlying spatial network framework can be easily adapted to study other plant diseases such as soybean rust, peanut leaf spot, and bacterial spot in tomato. Its explicit representation of field geometry makes it broadly applicable to spatially structured cropping systems and provides a transferable tool for evaluating disease-management strategies across diverse agricultural settings.

Building on this generalizable structure, future work could extend the SEIRB model to incorporate additional sources of field-level variation that influence disease spread. Weather conditions such as humidity, leaf wetness, and temperature-driven changes in pathogen activity could be integrated to improve predictive realism, while spatial factors such as soil texture variation, row geometry, and small-scale topography could refine representation of the environmental reservoir. Seasonal changes in canopy structure and plant development may also be incorporated to capture shifts in contact patterns and exposure pathways throughout the growing season. Another important direction is the inclusion of economic and operational constraints, enabling scouting schedules, fungicide applications, and removal strategies to be evaluated in terms of both biological impact and resource feasibility. By integrating realistic field geometry, dual transmission routes, and data-informed calibration, the proposed modeling framework provides mechanistic insight into the drivers of FLS spread and establishes a foundation for practical decision-making. More broadly, this approach supports the long-term goals of precision agriculture and sustainable crop protection in spatially structured cropping systems.

\appendix

\section*{Conflict of interest}
The authors declare that there is no conflict of interest.
\appendix

\section*{Reproducible code and scripts}\label{sec:codeFLS}
The code supporting this study is publicly available at \href{https://github.com/CPW93/FLS-Disease-Network-Model}{github.com/CPW93/FLS-Disease-Network-Model}.

\section*{Acknowledgment}
	This research work was partially supported by the NSF Award 2324691.

\section*{Appendix}\label{sec:appendix}

The following presents theoretical results for the equilibrium soil reservoir.

\subsection*{Theoretical Results}\label{sec:theory_result}

We formalize the relationship between infection pressure, management parameters, and the resulting soil reservoir dynamics at the plot level. The following lemma gives a closed-form equilibrium for $B_p$ and shows that the equilibrium inoculum increases with infection intensity $I_p$ and deposition $\xi$ and decreases with decay $\tau_{P(p)}$. 

\begin{lemma} \label{lem:1}
Let $r_{0}>0$, $k_{0}>0$, $\xi\ge 0$, $\tau_{P(p)}\ge 0$ and $I_p\ge 0$ be constants for a given plot $p$.
Consider the soil reservoir dynamics
\[
\dot B_p(t)=\big(r_{0}-\tau_{P(p)}\big)B_p(t)-\frac{r_{0}}{k_{0}}\,B_p(t)^2+\xi\,I_p ,
\]
where $\dot B_p(t)$ denotes the time derivative of $B_p(t)$, representing the instantaneous rate of change of the soil inoculum in plot $p$.

Then there is a unique nonnegative equilibrium $B_p^\star$ given by
\[
B_p^\star=\frac{k_{0}}{2r_{0}}\!\left[(r_{0}-\tau_{P(p)})+
\sqrt{\,\big(r_{0}-\tau_{P(p)}\big)^2+\frac{4r_{0}}{k_{0}}\,\xi\,I_p}\right].
\]
Moreover,
\[
\frac{\partial B_p^\star}{\partial I_p}
=\frac{\xi}{\sqrt{\,(r_{0}-\tau_{P(p)})^2+\tfrac{4r_{0}}{k_{0}}\,\xi\,I_p}}\;>\;0,\qquad
\frac{\partial B_p^\star}{\partial \xi}
=\frac{I_p}{\sqrt{\,(r_{0}-\tau_{P(p)})^2+\tfrac{4r_{0}}{k_{0}}\,\xi\,I_p}}\;>\;0,
\]
\[
\frac{\partial B_p^\star}{\partial \tau_{P(p)}}
=-\,\frac{k_{0}}{2r_{0}}\!\left(1+\frac{r_{0}-\tau_{P(p)}}{\sqrt{\,(r_{0}-\tau_{P(p)})^2+\tfrac{4r_{0}}{k_{0}}\,\xi\,I_p}}\right)\;<\;0.
\]

\end{lemma}

\begin{proof}
In the coupled host-environment system, the infection dynamics $I(t)$ typically evolve on a much faster time scale than the concentration of environmental pathogens $B(t)$. Therefore, for analytical tractability, we assume that $I(t)$ rapidly approaches a quasi-steady state relative to the slower evolution of $B(t)$. This allows us to treat $I(t)$ as approximately constant when analyzing the equilibrium and stability of the fungus dynamics. Such an assumption is known as time-scale separation or the quasi-steady-state approximation.

At equilibrium, $0=\big(r_{0}-\tau_{P(p)}\big)B-\frac{r_{0}}{k_{0}}B^2+\xi I_p = g(B)$, a concave quadratic in $B$ with
$g(0)=\xi I_p\ge 0$ and $g(B)\to -\infty$ as $B\to\infty$, so there is a unique nonnegative root.
Solving the quadratic yields the stated $B_p^\star$ (take the “positive” root to keep $B_p^\star\ge 0$).
Differentiating the closed form with respect to $I_p$, $\xi$, and $\tau_{P(p)}$ gives the displayed
partials; their signs follow since the square-root denominator is strictly positive for all parameters
in the stated ranges.

For the no-shedding special case $I_p=0$, the formula reduces to $B_p^\star=0$ if $r_{0}\le \tau_{P(p)}$ and
$B_p^\star=k_{0}(1-\tau_{P(p)}/r_{0})$ if $r_{0}>\tau_{P(p)}$, matching the logistic-with-decay limit.
Finally, if $\beta_{P(i)}>0$, the hazard term $\beta_{P(i)}B_{p(i)}$ is monotone in $B_{p(i)}$, so the same comparative statics apply.
\end{proof}

\noindent\textbf{Remark.}
Lemma~\ref{lem:1} tells us that higher infection levels and deposition rates lead to increased equilibrium soil inoculum, while longer decay times suppress it. This pattern is consistent with finding the soil-driven hazard, which rises with infection and deposition and falls with decay by \cite{BockusShroyer1998_ReducedTillage, Krupinsky2002_ManagingRisk}.

\printbibliography

@article{MBE2024FLS,
  author  = {Yang, Chayu and Wang, Jin},
  title   = {A mathematical model for frogeye leaf spot epidemics in soybean},
  journal = {Mathematical Biosciences and Engineering},
  year    = {2024},
  volume  = {21},
  number  = {1},
  pages   = {1144--1166},
  doi     = {10.3934/mbe.2024048},
  url     = {https://www.aimspress.com/article/doi/10.3934/mbe.2024048}
}

@article{KeelingEames2005,
  author  = {Keeling, Matt J. and Eames, Ken T. D.},
  title   = {Networks and epidemic models},
  journal = {Journal of the Royal Society Interface},
  year    = {2005},
  volume  = {2},
  number  = {4},
  pages   = {295--307},
  url = {https://doi.org/10.1098/rsif.2005.0051}
}

@article{Danon2011,
  author  = {Danon, Leon and Ford, Ashley P. and House, Thomas and Jewell, Chris P. and Keeling, Matt J. and Roberts, Gareth O. and Ross, Joshua V. and Vernon, Matt C.},
  title   = {Networks and the Epidemiology of Infectious Disease},
  journal = {Interdisciplinary Perspectives on Infectious Diseases},
  year    = {2011},
  pages   = {284909},
  doi     = {10.1155/2011/284909}
}

@article{Neves2022,
  author  = {Neves, D. L. and others},
  title   = {First detection of frogeye leaf spot in soybean fields in North Dakota and the G143A mutation in the cytochrome b gene of {Cercospora sojina}},
  journal = {Plant Health Progress},
  year    = {2022},
  volume  = {23},
  number  = {3},
  doi     = {10.1094/PHP-10-21-0132-BR}
}

@article{ShawPautasso2014,
  author  = {Shaw, M. W. and Pautasso, M.},
  title   = {Networks and Plant Disease Management: Concepts and Applications},
  journal = {Annual Review of Phytopathology},
  year    = {2014},
  volume  = {52},
  pages   = {477--493},
  doi     = {10.1146/annurev-phyto-102313-050229}
}

@article{Gilligan2008,
  author = {Gilligan, Christopher A},
    title = {Sustainable agriculture and plant diseases: an epidemiological perspective},
    journal = {Philosophical Transactions of the Royal Society B: Biological Sciences},
    volume = {363},
    number = {1492},
    pages = {741-759},
    year = {2007},
    month = {09},
    issn = {0962-8436},
    doi = {10.1098/rstb.2007.2181},
    url = {https://doi.org/10.1098/rstb.2007.2181},
    eprint = {https://royalsocietypublishing.org/rstb/article-pdf/363/1492/741/375058/rstb.2007.2181.pdf}
}

@article{Barro2023FLS,
  author  = {Barro, J. P. and Neves, D. L. and Del Ponte, and Emerson M. and others},
  title   = {Frogeye leaf spot caused by {Cercospora sojina}: A review},
  journal = {Tropical Plant Pathology},
  year    = {2023},
  volume  = {48},
  pages   = {363--374},
  doi     = {10.1007/s40858-023-00583-8}
}

@misc{NCStateFLS,
  author       = {{North Carolina State Extension}},
  title        = {Frogeye Leaf Spot of Soybean},
  year         = {2022},
  howpublished = {\url{https://soybeans.ces.ncsu.edu/}},
  note         = {Accessed 2025-09-25}
}

@misc{NDSU_FLS2023,
  author       = {{North Dakota State University Extension}},
  title        = {Soybean Disease Diagnostic Series},
  howpublished = {\url{https://www.ndsu.edu/agriculture/extension/publications/soybean-disease-diagnostic-series}},
  note         = {Accessed 2025-09-26}
}

@misc{ISU_FLS,
  author    = {Stoetzer, Ethan},
  title     = {Frogeye Leaf Spot Making an Appearance Across Iowa Soybeans},
  year      = {2018},
  url       = {https://crops.extension.iastate.edu/cropnews/2018/08/frogeye-leaf-spot-making-appearance-across-iowa-soybeans},
  note      = {Integrated Crop Management News, Iowa State University Extension and Outreach; Accessed 2025-09-26}
}

@article{Allen2017,
  author  = {Allen, Tom W. and Bradley, Carl A. and Sisson, Adam J. and Byamukama, Emmanuel and Chilvers, Martin I. and Coker, Cliff M. and Collins, Alyssa A. and Damicone, John P. and Dorrance, Anne E. and Dufault, Nicholas S. and Esker, Paul D. and Faske, Travis R. and Giesler, Loren J. and Grybauskas, Arvydas P. and Hershman, Donald E. and Hollier, Clayton A. and Isakeit, Tom and Jardine, Douglas J. and Kelly, Heather M. and Kemerait, Robert C. and Kleczewski, Nathan M. and Koenning, Steve R. and Kurle, James E. and Malvick, Dean K. and Markell, Samuel G. and Mehl, Hillary L. and Mueller, Daren S. and Mueller, John D. and Mulrooney, Robert P. and Nelson, Berlin D. and Newman, Melvin A. and Osborne, Larry and Overstreet, Charles and Padgett, G. Boyd and Phipps, Patrick M. and Price, Paul P. and Sikora, Edward J. and Smith, Damon L. and Spurlock, Terry N. and Tande, Connie A. and Tenuta, Albert U. and Wise, Kiersten A. and Wrather, J. Allen},
  title   = {Soybean Yield Loss Estimates Due to Diseases in the United States and Ontario, Canada, from 2010 to 2014},
  journal = {Plant Health Progress},
  year    = {2017},
  volume  = {18},
  number  = {1},
  pages   = {19--27},
  doi     = {10.1094/PHP-RS-16-0066}
}

@article{Bradley2021,
  author  = {Bradley, Carl A. and Allen, Tom W. and Sisson, Adam J. and Bergstrom, Gary C. and Bissonnette, Kaitlyn M. and Bond, Jason and Byamukama, Emmanuel and Chilvers, Martin I. and Collins, Alyssa A. and Damicone, John P. and Dorrance, Anne E. and Dufault, Nicholas S. and Esker, Paul D. and Faske, Travis R. and Fiorellino, Nicole M. and Giesler, Loren J. and Hartman, Glen L. and Hollier, Clayton A. and Isakeit, Tom and Jackson-Ziems, Tamra A. and Jardine, Douglas J. and Kelly, Heather M. and Kemerait, Robert C. and Kleczewski, Nathan M. and Koehler, Alyssa M. and Kratochvil, Robert J. and Kurle, James E. and Malvick, Dean K. and Markell, Samuel G. and Mathew, Febina M. and Mehl, Hillary L. and Mehl, Kelsey M. and Mueller, Daren S. and Mueller, John D. and Nelson, Berlin D. and Overstreet, Charles and Padgett, G. Boyd and Price, Paul P. and Sikora, Edward J. and Small, Ian and Smith, Damon L. and Spurlock, Terry N. and Tande, Connie A. and Telenko, Darcy E. P. and Tenuta, Albert U. and Thiessen, Lindsey D. and Warner, Fred and Wiebold, William J. and Wise, Kiersten A.},
  title   = {Soybean Yield Loss Estimates Due to Diseases in the United States and Ontario, Canada, from 2015 to 2019},
  journal = {Plant Health Progress},
  year    = {2021},
  volume  = {22},
  number  = {4},
  pages   = {483--495},
  doi     = {10.1094/PHP-01-21-0013-RS}
}

@incollection{Wise2015,
  author    = {Hartman, Glen L. and Rupe, John C. and Sikora, Edward J. and Domier, Leslie L. and Davis, Jeffrey A. and Steffey, Kevin L.},
  title     = {Soybean Disease and Pest Management Strategies},
  booktitle = {Compendium of Soybean Diseases and Pests},
  
  edition   = {5},
  year      = {2016},
  pages     = {167--173},
  address   = {St. Paul, MN},
  publisher = {American Phytopathological Society},
  isbn      = {978-0-89054-475-4},
  note      = {eBook DOI (book-level): 10.1094/9780890544754},
  url = {https://doi.org/10.1094/9780890544754}
}

@article{SIRModel,
  author  = {Wang, Bingqing and Hu, Yanyao and Cen, Zhongdi and Huang, Jian and He, Tianfeng and Xu, Aimin},
  title   = {A {SIR} Model with Incomplete Data for the Analysis of Influenza A Spread in Ningbo},
  journal = {Discrete Dynamics in Nature and Society},
  year    = {2024},
  volume  = {2024},
  number  = {1},
  pages   = {7694770},
  doi     = {10.1155/2024/7694770}
}

@book{Murray2003,
  author    = {Murray, James D.},
  title     = {Mathematical Biology II: Spatial Models and Biomedical Applications},
  edition   = {3},
  year      = {2003},
  series    = {Interdisciplinary Applied Mathematics},
  volume    = {18},
  publisher = {Springer},
  address   = {New York},
  doi       = {10.1007/b98869},
  url       = {https://link.springer.com/book/10.1007/b98869}
}

@article{Garrett2018Networks,
  author  = {Garrett, Karen A. and Alcal{\'a}-Brise{\~n}o, Ricardo I. and Andersen, Katherine F. and Buddenhagen, Christine E. and Choudhury, Rishi A. and Fulton, Jeffrey C. and Hernandez Nopsa, Jorge F. and Poudel, Rachana and Xing, Yanru},
  title   = {Network Analysis: A Systems Framework to Address Grand Challenges in Plant Pathology},
  journal = {Annual Review of Phytopathology},
  year    = {2018},
  volume  = {56},
  pages   = {559--580},
  doi     = {10.1146/annurev-phyto-080516-035326}
}

@article{Mengistu2014FLS,
  author  = {Mengistu, Alemu and Kelly, Heather M. and Bellaloui, Nacer and Arelli, Prakash R. and Reddy, Krishna N. and Wrather, Allen J.},
  title   = {Tillage, Fungicide, and Cultivar Effects on Frogeye Leaf Spot Severity and Yield in Soybean},
  journal = {Plant Disease},
  year    = {2014},
  volume  = {98},
  number  = {11},
  pages   = {1476--1484},
  doi     = {10.1094/PDIS-12-13-1268-RE}
}

@article{Pudlo2016,
  author  = {Pudlo, P. and Marin, J.-M. and Estoup, A. and Cornuet, J.-M. and Gauthier, M. and Robert, C. P.},
  title   = {Reliable {ABC} model choice via random forests},
  journal = {Bioinformatics},
  year    = {2016},
  volume  = {32},
  number  = {6},
  pages   = {859--866},
  url = {https://doi.org/10.1093/bioinformatics/btv684}
}

@article{Raynal2019,
  author  = {Raynal, L. and Marin, J.-M. and Pudlo, P. and Ribatet, M. and Robert, C. P. and Estoup, A.},
  title   = {{ABC} random forests for {B}ayesian parameter inference},
  journal = {Bioinformatics},
  year    = {2019},
  volume  = {35},
  number  = {10},
  pages   = {1720--1728},
  url = {https://doi.org/10.1093/bioinformatics/bty867}
}

@article{DrovandiPettitt2011_SMC_ABC,
  author  = {Drovandi, Christopher C. and Pettitt, Anthony N.},
  title   = {Estimation of Parameters for Macroparasite Population Evolution Using Approximate Bayesian Computation},
  journal = {Biometrics},
  year    = {2011},
  volume  = {67},
  number  = {1},
  pages   = {225--233},
  doi     = {10.1111/j.1541-0420.2010.01410.x}
}

@incollection{Nleya2019SoybeanStages,
  author    = {Nleya, Thandiwe and Sexton, Peter and Gustafson, Kyle and Moriles Miller, Janet},
  title     = {Soybean Growth Stages},
  booktitle = {iGrow Soybeans: Best Management Practices},
  publisher = {SDSU Extension, South Dakota State University},
  year      = {2019},
  chapter   = {3},
  pages     = {3-25--3-34},
  note      = {Extension handbook chapter},
  url = {https://extension.sdstate.edu/sites/default/files/2020-03/S-0004-03-Soybean.pdf}
}

@article{BockusShroyer1998_ReducedTillage,
  author  = {Bockus, W. W. and Shroyer, J. P.},
  title   = {The Impact of Reduced Tillage on Soilborne Plant Pathogens},
  journal = {Annual Review of Phytopathology},
  year    = {1998},
  volume  = {36},
  number  = {1},
  pages   = {485--500},
  doi     = {10.1146/annurev.phyto.36.1.485},
  url     = {https://doi.org/10.1146/annurev.phyto.36.1.485},
  issn    = {0066-4286},
  eissn   = {1545-2107},
  pmid    = {15012510},
  month   = aug
}

@article{Krupinsky2002_ManagingRisk,
  author    = {Krupinsky, Joseph M. and Bailey, Karen L. and McMullen, Marcia P. and Gossen, Bruce D. and Turkington, T. Kelly},
  title     = {Managing Plant Disease Risk in Diversified Cropping Systems},
  journal   = {Agronomy Journal},
  year      = {2002},
  volume    = {94},
  number    = {2},
  pages     = {198--209},
  doi       = {10.2134/agronj2002.1980},
  url       = {https://doi.org/10.2134/agronj2002.1980},
  issn      = {0002-1962},
  eissn     = {1435-0645},
  publisher = {American Society of Agronomy}
}

@article{PDE_Jin_2025,
author = {Anderson, M. and Yamazaki, K. and Yang, C. and Wang, J.},
title = {A partially diffusive model for frogeye leaf spot epidemics in soybean},
journal = {Discrete and Continuous Dynamical Systems - B},
year = {2025},
issn = {1531-3492},
doi = {10.3934/dcdsb.2025177},
url = {https://www.aimsciences.org/article/id/691440fb6547bb6e5b3da2dc}
}

@article{FARMAN,
author ={Farman, M. and Ullah, I. and Hincal, E. and Nisar, K. S. and Hosseini, K. and Sambas, A.},
title = {Stability and dynamics of fractional order frogeye leaf spot infection model with fungal density function},
journal = {Computational Biology and Chemistry},
year = {2025},
volume = {120},
issn = {1476-9271},
doi = {},
url = {https://www.sciencedirect.com/science/article/pii/S1476927125003858},

}

@article{agronomy_machine_learning,
title = {HSDT-TabNet: A Dual-Path Deep Learning Model for Severity Grading of Soybean Frogeye Leaf Spot},
journal = {Agronomy},
volume = {15},
pages = {108724},
year = {2025},
issn = {2073-4395},
url ={https://www.mdpi.com/2073-4395/15/7/1530},
author = {Li, X. and Zhou, Y. and Li, Y. and Wang, S. and Bian, W. and Sun, H.},
}

@article{agronomy_FLS_Seviarity,
author = {Phillips, X. A. and Kandel, Y. R. and Mueller, D. S.},
title = {Impact of Foliar Fungicides on Frogeye Leaf Spot Severity, Radiation Use Efficiency and Yield of Soybean in Iowa},
journal = {Agronomy},
volume = {11},
year = {2021},
url = {https://www.mdpi.com/2073-4395/11/9/1785},
issn = {2073-4395},
doi = {https://doi.org/10.3390/agronomy11091785}
}

@article{Franks01102020,
author = {Jordan J. Franks},
title = {Handbook of Approximate Bayesian Computation.},
journal = {Journal of the American Statistical Association},
volume = {115},
number = {532},
pages = {2100--2101},
year = {2020},
publisher = {ASA Website},
doi = {10.1080/01621459.2020.1846973},


URL = { 
        https://doi.org/10.1080/01621459.2020.1846973
},
eprint = { 
        https://doi.org/10.1080/01621459.2020.1846973
}

}

@article{Cai2024FLS,
author = {Bandara, Ananda Y., and Weerasooriya, Dilooshi K., and Bradley, Carl A., and Allen, Tom W., and Esker, Paul D.},
title = {Dissecting the economic impact of soybean diseases in the United States over two decades},

note = {Publisher Copyright: {\textcopyright} 2020 Bandara et al. This is an open access article distributed under the terms of the Creative Commons Attribution License, which permits unrestricted use, distribution, and reproduction in any medium, provided the original author and source are credited.},
year = {2020},
doi = {10.1371/journal.pone.0231141},
language = {English (US)},
volume = {15},
journal = {PloS one},
issn = {1932-6203},
publisher = {Public Library of Science},
number = {4},
}

@article{Neves2020QoI,
  author = {Neves, Danilo L. and Chilvers, Martin I. and Jackson-Ziems, Tamra A. and Malvick, Dean K. and Bradley, Carl A.},
title = {Resistance to Quinone Outside Inhibitor Fungicides Conferred by the G143A Mutation in Cercospora sojina (Causal Agent of Frogeye Leaf Spot) Isolates from Michigan, Minnesota, and Nebraska Soybean Fields},
journal = {Plant Health Progress},
volume = {21},
number = {4},
pages = {230-231},
year = {2020},
doi = {10.1094/PHP-06-20-0052-BR},

URL = { 
    
        https://doi.org/10.1094/PHP-06-20-0052-BR
    
    

},
eprint = { 
    
        https://doi.org/10.1094/PHP-06-20-0052-BR
    
    

}

}

@article{Neves2021Baseline,
  title = "Sensitivity of Cercospora sojina to demethylation inhibitor and methyl benzimidazole carbamate fungicides",
author = "Guirong Zhang and Neves, \{Danilo L.\} and Kelsey Krausz and Bradley, \{Carl A.\}",
note = "Publisher Copyright: {\textcopyright} 2021 Elsevier Ltd",
year = "2021",
month = nov,
doi = "10.1016/j.cropro.2021.105765",
language = "English",
volume = "149"
}

@article{Bradley2023INCCA,
  author  = {Barro, Jhonatan P. and Del Ponte, Emerson M. and Allen, Tom W. and Bond, Jason P. and Faske, Travis R. and Hollier, Clayton A. and Kandel, Yuba R. and Mueller, Daren S. and Kelly, Heather M. and Kleczewski, Nathan M. and Ames, Keith A. and Price, Paul P. and Sikora, Edward J. and Bradley, Carl A.},
  title   = {Efficacy and Profitability of Fungicides for Managing Frogeye Leaf Spot on Soybean in the United States: A 10-Year Quantitative Summary},
  journal = {Plant Disease},
  year    = {2023},
  doi     = {10.1094/PDIS-02-23-0291-RE},
  note    = {Published online 26 October 2023}
}

@article{Dashiell1991,
  author       = {Dashiell, K. E. and Akem, C. N.},
  title        = {Yield losses in soybeans from frogeye leaf spot caused by \textit{Cercospora sojina}},
  journal      = {Crop Protection},
  year         = {1991},
  volume       = {10},
  number       = {6},
  pages        = {465--468},
  doi          = {10.1016/S0261-2194(91)80134-2}
}

@article{KermackMcKendrick1927,
  author  = {Kermack, W. O. and McKendrick, A. G.},
  title   = {A Contribution to the Mathematical Theory of Epidemics},
  journal = {Proceedings of the Royal Society of London. Series A, Containing Papers of a Mathematical and Physical Character},
  volume  = {115},
  number  = {772},
  pages   = {700--721},
  year    = {1927},
  doi     = {10.1098/rspa.1927.0118}
}

@misc{CPN2024,
  author = {{Crop Protection Network}},
  title  = {Soybean Disease Loss Estimates from the United States and Ontario, Canada -- 2024},
  year   = {2024},
  url    = {https://loss.cropprotectionnetwork.org/},
  note   = {Accessed September 2024}
}

@techreport{Webster2026NDReport,
  author      = {Webster, Wade and Dusek, Gabe and Becton, Hope},
  title       = {2025 Extension Soybean Pathology Field Research Reports},
  institution = {North Dakota State University Extension},
  year        = {2026},
  number      = {PP2290},
  url         = {https://www.ndsu.edu/agriculture/extension/publications/2025-extension-soybean-pathology-field-research-reports}
}

@article{Zhang2012,
  author  = {Zhang, G. R. and Newman, M. A. and Bradley, C. A.},
  title   = {First report of the soybean frogeye leaf spot pathogen \textit{Cercospora sojina} resistant to quinone outside inhibitor fungicides in the United States},
  journal = {Plant Disease},
  year    = {2012},
  volume  = {96},
  number  = {5},
  pages   = {767},
  doi     = {10.1094/PDIS-10-11-0915-PDN}
}

@article{Cai2024IndianaFLS,
  author  = {Cai, Guohong and Lopes da Silva, Leandro and Pi{\~n}eros-Guerrero, Natalia and Telenko, Darcy E. P.},
  title   = {Population Structure and Mating Type Distribution of {\textit{Cercospora sojina}} from Soybeans in Indiana, United States},
  journal = {Journal of Fungi},
  year    = {2024},
  volume  = {10},
  number  = {11},
  pages   = {802},
  doi     = {10.3390/jof10110802},
  url     = {https://doi.org/10.3390/jof10110802}
}

@article{Neves2022ND,
  author  = {Neves, Danilo L. and Berghuis, Brandt G. and Halvorson, Jessica M. and Hansen, Bryan C. and Markell, Samuel G. and Bradley, Carl A.},
  title   = {First Detection of Frogeye Leaf Spot in Soybean Fields in North Dakota and the \textit{G143A} Mutation in the Cytochrome \textit{b} Gene of \textit{Cercospora sojina}},
  journal = {Plant Health Progress},
  year    = {2022},
  doi     = {10.1094/PHP-10-21-0132-BR},
  url     = {https://doi.org/10.1094/PHP-10-21-0132-BR}
}

\end{document}